\def\year{2020}\relax
\documentclass[letterpaper]{article} 
\usepackage{aaai20}  
\usepackage{times}  
\usepackage{helvet} 
\usepackage{courier}  
\usepackage[hyphens]{url}  
\usepackage{graphicx} 
\usepackage{graphicx}  
\usepackage{soul}
\usepackage{url}
\usepackage[utf8]{inputenc}

\usepackage{booktabs}
\usepackage{mathrsfs}
\usepackage{ntheorem}
\usepackage{graphicx}
\usepackage{subcaption}
\newtheorem{thm}{Theorem}
\newtheorem{mydef}{Definition}
\newtheorem{lem}{Lemma}
\newtheorem*{proof}{Proof}
\newtheorem{assumption}{Assumption}
\newcommand{\citet}[1]{\citeauthor{#1} \shortcite{#1}}

\usepackage[linesnumbered,ruled,vlined]{algorithm2e}
\usepackage{algpseudocode}
\usepackage{amsmath}

\title{Mechanism Design with Predicted Task Revenue for Bike Sharing Systems}
\author{ \Large \textbf{Hongtao Lv,\textsuperscript{\rm 1} Chaoli Zhang,\textsuperscript{\rm 1} Zhenzhe Zheng,\textsuperscript{\rm 1} Tie Luo,\textsuperscript{\rm 2} Fan Wu,\textsuperscript{\rm 1}\thanks{F. Wu is the corresponding author.} Guihai Chen\textsuperscript{\rm 1}}\\ 
\textsuperscript{\rm 1} Department of Computer Science and Engineering, Shanghai Jiao Tong University, China\\
\textsuperscript{\rm 2} Department of Computer Science, Missouri University of Science and Technology, USA\\
\{lvhongtao, chaoli\_zhang, zhengzhenzhe\}@sjtu.edu.cn, tluo@mst.edu, \{fwu, gchen\}@cs.sjtu.edu.cn
}
\begin{document}

\maketitle

\begin{abstract}

Bike sharing systems have been widely deployed around the world in recent years. A core problem in such systems is to reposition the bikes so that the distribution of bike supply is reshaped to better match the dynamic bike demand. When the bike-sharing company or platform is able to predict the revenue of each reposition task based on historic data, an additional constraint is to cap the payment for each task below its predicted revenue. In this paper, we propose an incentive mechanism called {\em TruPreTar} to incentivize users to park bicycles at locations desired by the platform toward rebalancing supply and demand. TruPreTar possesses four important economic and computational properties such as truthfulness and budget feasibility. Furthermore, we prove that even when the payment budget is tight, the total revenue still exceeds or equals the budget. Otherwise, TruPreTar achieves 2-approximation as compared to the optimal (revenue-maximizing) solution, which is close to the lower bound of at least $\sqrt{2}$ that we also prove. Using an industrial dataset obtained from a large bike-sharing company, our experiments show that TruPreTar is effective in rebalancing bike supply and demand and, as a result, generates high revenue that outperforms several benchmark mechanisms.

\end{abstract}

\section{Introduction}\label{sec:intro}

Bike sharing is a new transportation mode with many benefits in offering convenience and flexibility as well as lowering economic cost. By 2015, more than 7000 bike sharing systems have been deployed around the world \cite{laporte2015shared}. However, the flexibility of bike sharing systems, in particular the ``anywhere-parking'' convenience, brings forth a serious issue of imbalance between the distribution of bike supply and demand. This leads to many users being unable to find a bicycle nearby when they need it, and ultimately affects company revenue adversely. Hence, there is an urgent need to rebalance the supply and demand by repositioning the bicycles, which we refer to as a {\em bike rebalancing problem}.

There are two approaches to solving this problem. One is to relocate bicycles by the staff of the bike-sharing company or platform, for example using trucks. This involves route planning and typically uses linear programming techniques, as has been extensively studied under static models \cite{maggioni2019stochastic,schuijbroek2017inventory} and dynamic models \cite{kek2009decision,angeloudis2014strategic}. However, this approach is costly and not eco-friendly in terms of carbon footprint.

Another approach is to design incentive mechanisms to motivate users to reposition bicycles at platform-desired locations for rebalancing supply and demand. This falls under the research area of crowdsourcing \cite{fricker2016incentives} which takes advantage of the power of crowd to complete tasks that are otherwise difficult \cite{luo2016incentive}.
With this idea, \citet{ghosh2017incentivizing} proposed a solution that generates repositioning tasks together with the use of bike trailers, and pays users using the Vickrey-Clarke-Groves (VCG) mechanism. \citet{singla2015incentivizing} introduced a crowdsourcing method that offers users monetary incentive for parking their bikes at recommended locations.

However, most of existing mechanisms for this problem have not considered the predicted revenue (or value) of a repositioning task. Such predicted values are made available in recent years, due to the prosperity of deep learning technique which can well predict the bike demand in such systems \cite{yang2016mobility,li2019citywide,Zhang2016DNN}. With this predictive power, the expected revenue or value of a repositioning task can be easily obtained. Therefore, the platform would assign a task to a user only if the payment to the user does not exceed the predicted revenue. This constraint is akin to the ``reserve price" in forward auctions which are used to promote revenue. In our problem, it is used to better control the payments and thereby increase the profit of the platform.
This constraint is generally overlooked in prior work on bike sharing or crowdsourcing \cite{angelopoulos2018incentivized,goel2013matching,vaze2017online}. For instance with classic budget feasible mechanisms such as \cite{singer2010budget}, the payment to a winning user is determined by the bids of other users and/or the budget. But in our case, this payment is invalid if it is above the value of the task assigned to the winning user.

In this paper, we propose a solution called \underline{Tru}thful and budget-feasible incentive mechanism with \underline{Pre}dicted \underline{Ta}sk \underline{r}evenue (TruPreTar) for the bike rebalancing problem, which we model as a reverse auction. In this auction, users bid for repositioning tasks and the platform determines task allocation and user payments. We first show that several widely used auctions and pricing mechanisms do not apply to this problem. Then we present the design of TruPreTar and show that it is an effective solution that satisfies several desired economic and computational properties. In addition, we prove two important theoretical guarantees backed by our mechanism on revenue. Finally, we evaluate the performance of TruPreTar via simulations using a real industrial dataset obtained from a large bike-sharing company.

Our main contributions are summarized as follows:
\begin{itemize}
\item We model the bike rebalancing problem under a crowdsourcing framework using a reverse auction model.
    Importantly, we incorporate a bipartite graph into the auction---with a payment constraint---to determine the allocation rule and the payment rule for the auction. The payment constraint also implies the establishment of a connection between (bike) demand prediction and incentive mechanism design.
\item We propose a mechanism TruPreTar which satisfies desired properties including incentive compatibility, budget feasibility, individual rationality, and computational efficiency. Notably, we achieve this by novelly combining Myerson's Lemma and a greedy weighted maximum matching technique.
\item We prove two important theoretical guarantees for TruPreTar: When the budget is tight, it ensures that the platform revenue is no less than its budget, under a practical large-market assumption. When the budget is sufficient, TruPreTar achieves a 2-approximation ratio as compared to the optimal solution which maximizes revenue. We show that the lower bound of the approximation ratio is at least $\sqrt{2}$ which means our ratio is rather close. Putting in practice perspectives, we provide a guideline of how to use this result to set company budget in real dynamic systems.
\item We evaluate the effectiveness of our mechanism via extensive experiments using a real dataset from a large bike-sharing company. The results show that our mechanism outperforms other benchmark mechanisms in terms of revenue and profit.
\end{itemize}

\section{Model}

Consider a dynamic bike sharing system in which a batch of $n$ users $N=\{1,2,\dots,n\}$ have hired their bicycles and not parked them yet. There is also a set of $m$ discrete bicycle parking lots or locations\footnote{We do not consider continuous locations because many countries such as Singapore and China have nowadays stipulated municipal regulations to restrict shared bikes to be parked at designated locations rather than arbitrarily.} $L=\{1,2,\dots,m\}$, where each location can accommodate multiple bicycles but is conceptually considered a point on a map. We assume that each user's destination is known to the platform when she hires a bicycle, since it can be either reported by the user, as is the practice adopted by some companies (e.g., HelloBike in China), or predicted using historic information \cite{liu2018two}.  The platform aims to incentivize the users to park their bicycles at system-desired locations for rebalancing the supply and demand of bicycles. Each user has a maximum relocation range $h$, out of which they would not accept a repositioning task; in other words, $h$ is the maximum ``extra mile'' they are willing to relocate.

Similar to the concept of \emph{first come first served flow} in \cite{waserhole727040vehicle}, we assume that a bicycle that is parked earlier will have a higher probability to be hired than a bicycle parked later at the same location. As such, each bicycle has a different probability of being hired. Accordingly, we define a repositioning task $t_{lx}$ as ``park your bicycle at location $l$ as the $x$-th bicycle", and it is associated with an expected revenue $r_{lx}$. As mentioned earlier, $r_{lx}$ can be derived from bike demand predication at the location \cite{yang2016mobility,li2019citywide}.

To assign the tasks to users, we construct a bipartite graph $G=\{N,T,E\}$, where the left nodes are the set of users $N$, and the right nodes are the set of tasks $T$. Since there are maximum $n$ possible tasks at each location, there are totally $m\times n$ tasks in $T$. The set $E$ is the edges between users and tasks indicating the {\em feasibility} of assignment: an edge between $i$ and $t_{lx}$ means that location $l$ is within the maximum relocation range of user $i$. For notation simplicity, henceforth we use $j$ to denote a task $t_{lx}$ when there is no ambiguity.

Each user $i$ has a relocation cost $c_i$ which is a private value known to user $i$ only. She bids for a task and claims a cost $b_i$ which is not necessarily equal to $c_i$. Each task $j$ has an expected revenue (or ``value" as we use interchangeably) $r_j$ for the platform. We aim to design an incentive mechanism that consists of an allocation rule and a payment rule, where the allocation rule specifies which task is allocated to which user, i.e., the matched pairs $(i,j)$, and the payment rule specifies a payment $p_i$ for each matched user $i$. The utility of a user is defined as $p_i - c_i$. In addition, the platform has a budget $B$ for all the rebalancing tasks.
We want our mechanism to satisfy the following desirable properties:
\begin{itemize}
\item \emph{Incentive Compatibility}: a user can maximize her utility only by bidding truthfully, i.e., $b_i=c_i$. This property is also known as \emph{truthfulness} or \emph{strategy-proofness}.
\item \emph{Individual Rationality} for both users and platform: the payment to each winning (i.e., matched) user should be no less than her cost, i.e., $p_i\ge c_i$; the payment for each matched task should also be no more than its value, i.e., $p_i \le r_j$ ($r_j$ is similar to the \emph{reserve price} in the mechanism design theory).
\item \emph{Budget Feasibility}: the overall payment should be no more than the budget, i.e., $\sum_{(i,j)\in M}p_i \le B$.
\item \emph{Computational Efficiency}: the mechanism should terminate in polynomial time.
\end{itemize}

Our objective is to maximize the platform revenue $R=\sum_{(i,j)\in M}r_j$, where $M=\{(i,j)\}$ is the set of matched user-task pairs.
Thus, the problem of revenue maximization via task allocation can be formulated as
\begin{align}\label{eq:problem}
 {\bf max }\quad R = &\sum_{(i,j)\in M} r_j\\
  {\bf s.t. } \quad p_i \le &r_j, \quad \forall (i,j) \in M  \notag \\
   p_i \ge &b_i, \quad \forall (i,j) \in M  \notag \\
   \sum_{(i,j)\in M} &p_i \le B \notag
\end{align}

In addition to revenue, we also evaluate platform profit which is defined as
\[ Pr=\sum_{(i,j)\in M}r_j-\sum_{(i,j)\in M}p_i. \]

\subsection{Infeasibility of Existing Mechanisms}

In this section, we show that some widely used mechanisms are not feasible for the bike rebalancing problem.

\subsubsection{VCG mechnism.} This is a classical mechnism that is strategy-proof and maximizes social welfare. However, it does not guarantee budget feasibility as required in our case. Proof by counter-example: see Figure \ref{fig:a} (a), where the platform has budget $1$ and there are two users $\{a, b \}$ both with a small cost $\epsilon$; the two tasks $\{1, 2 \}$ both have value $1$ and there are two edges $\{(a,1),(b,2)\}$. The VCG mechanism will output matching $M=\{(a,1),(b,2)\}$ and the payment for each user is 1. 
Thus, the overall payment exceeds the budget and hence the mechanism is not budget feasible.

\begin{figure}[h]
\centering
\subcaptionbox{$B=1$}{\includegraphics[width=0.49\columnwidth]{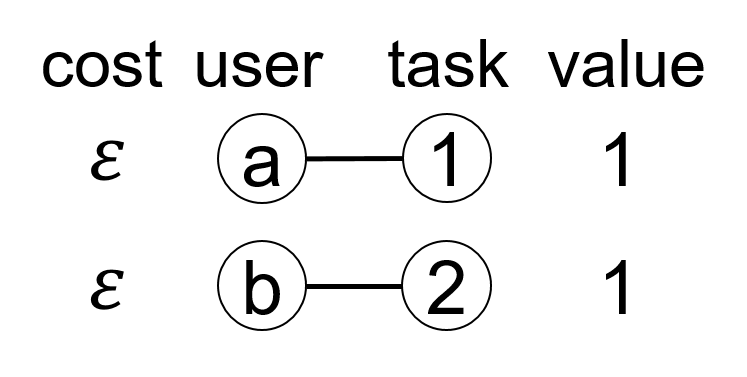}}
\hfill
\subcaptionbox{$B\gg5$}{\includegraphics[width=0.49\columnwidth]{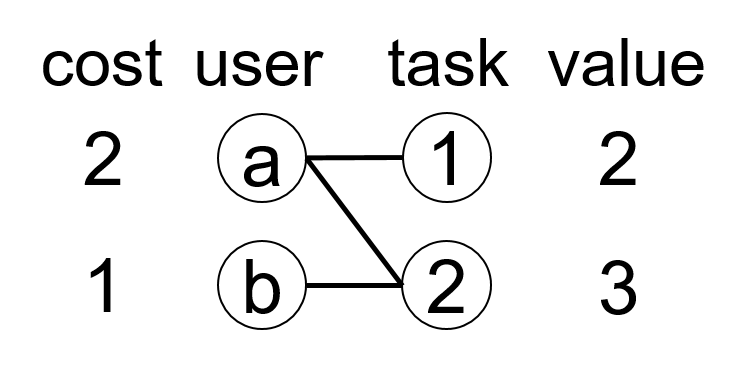}}
\vfill
\subcaptionbox{$B\gg6$}{\includegraphics[width=0.49\columnwidth]{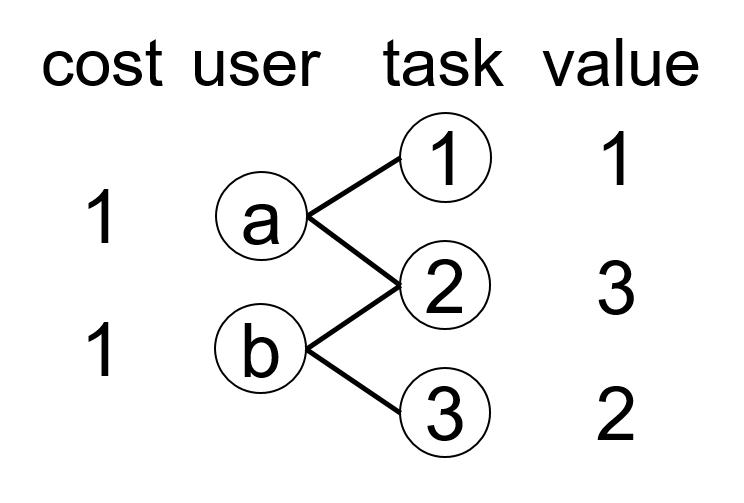}}
\caption{Counter examples for existing mechanisms.}
\label{fig:a}
\end{figure}

\subsubsection{Singer's mechanism.}
Proposed in \cite{singer2010budget}, this is another well-known mechanism yet is budget feasible. It greedily allocates user-task pairs with the highest ratio of $r_j/b_i$, but does not consider the individual rationality of the platform (i.e., $p_i\le r_j$). If we adapt the mechanism by adding the constraint into the payment rule, it will no longer guarantee incentive compatibility. To illustrate this, see Figure \ref{fig:a} (b), the platform has sufficient budget, and there are two users $\{a, b \}$ with costs $c_a=2,c_b=1$, two tasks $\{1, 2\}$ with values $r_1=2, r_2=3$, and three edges $\{(a,1),(a,2),(b,2)\}$. The adapted mechanism outputs match $M=\{(a,1),(b,2)\}$, and $p_a=2$. However, if user $a$ misreports her cost to be a small number $\epsilon$, she will be allocated task $2$ and her payment becomes $3$, making her better off.

\subsubsection{Optimal matching.}
The third method is to choose a subgraph with $\sum r_j\le B$ and then use the optimal (i.e., maximum) matching (that maximizes platform revenue) to allocate tasks, where the payments to winners are set to the values of their matched tasks. Not only does this mechanism has zero profit for the platform, but more importantly, it is also not truthful. To see this, consider Figure \ref{fig:a} (c), where the budget is sufficient, and there are two users $\{a, b \}$ with costs $c_a=1,c_b=1$, three tasks $\{1, 2, 3\}$ with values $r_1=1, r_2=3, r_3=2$, and four edges $\{(a,1),(a,2),(b,2),(b,3)\}$. The optimal matching will output match $M=\{(a,2),(b,3)\}$ since it achieves maximized revenue and the payments are $p_a=3$ and $p_b=2$. However, if user $b$ untruthfully bids a cost of 3 instead of 1, then she will be assigned task 2 instead of 3 because task value must be no less than the cost. Therefore, she will receive a higher payment than bidding truthfully. Therefore, incentive compatibility is violated.

In fact, we will prove in Theorem \ref{thm:lowerbd} that there is no optimal mechanism that can satisfy the four properties simultaneously. Therefore, inspired by \cite{zhang2018efficient}, we propose an approximate mechanism (i.e., TruPreTar) in this paper.

\section{Mechanism Design of TruPreTar}

In this section, we present our proposed mechanism TruPreTar for the bike rebalancing problem. We first introduce a notion called {\em right-perfect matching} in a bipartite graph.

\begin{mydef}
A right-perfect matching in a bipartite graph $G$ is a matching with size $|T|$, where $T$ is the set of right nodes in graph $G$.
\end{mydef}

In other words, we say a bipartite graph has a right-perfect matching if all tasks in the graph on the right can be matched to a user on the left. It is easy to see that a right-perfect matching is also a maximum matching of a bipartite graph.

\begin{algorithm}[!t]
\caption{TruPreTar: a truthful and budget feasible incentive mechanism with predicted task revenue}
\label{alg1}
\KwIn{Bipartite graph $G=(N, T, E)$, budget $B$, cost $c_i$, value $r_j$, $\forall i \in N,\forall j \in T$.}
\KwOut{Task allocation $M=\{(i, j)\}$ and payment $p_i$ for each winning user $i$.}
Let $F=N\cup T$. For an element $e\in F$, if $e$ is a user $i$, the value $v_e$ is defined as her bid $b_i$; if $e$ is a task $j$, the value $v_e$ is defined as the expected revenue $r_j$.

Delete all edges $(i,j)$ with $b_i>r_j$ in $G$.

$M \leftarrow \emptyset$, $B'\leftarrow B$, $T'\leftarrow \emptyset$, $N'\leftarrow \emptyset$, $E'\leftarrow \emptyset$, $G'=(N',T',E')$;

Sort elements in $F$ in decreasing order of $v_e$, breaking ties randomly, but if the tie is between a task and a user, let the task go first.

\For {each element $e$ in the above order}{
\If{$e$ is a task $j$}{
Let $N_j$ be the set of unmatched users connected to $j$, $E_j$ be the set of edges between $j$ and $N_j$.

\If{$G' \cup E_j$ has a right-perfect matching and \\ $(|T'|+1)\cdot r_j\le B'$}{
$G'\leftarrow G'\cup (N_j,{j},E_j)$.

$P\leftarrow r_j$.

\Else{
Skip to next element.
}}}

\If{$e$ is a user $i$ in $G'$}{
\If {$G'\backslash i$ has a right-perfect matching}{
$G'\leftarrow G'\backslash i$.

$P\leftarrow b_i$.
}}
\For{Each user $i\in G'$}{\label{checkstart}
\If {$G'\backslash i$ doesn't have a right-perfect matching}{
\For{Each edge $(i,j)$ of $i$ in $G'$}{
\If {$G'\backslash i\cup (i,j)$ has a right-perfect matching}{
$M\leftarrow M\cup (i,j)$, $p_i\leftarrow P$.

$B'\leftarrow B'-p_i.$

$G'\leftarrow G'\backslash\{i,j\}$.

Skip to next user.

\Else{
Skip to next edge.
\label{checkend}
}}}}
}
}
\end{algorithm}

The key idea of TruPreTar is to maintain a subgraph $G'$ that always has a right-perfect matching. TruPreTar sorts all the tasks and users together in decreasing order of their values (costs), and then iterates over this sorted list. The mechanism tries to include more tasks with high values in $G'$ and delete more users with high costs from it until the budget is exhausted or all elements are processed.

The complete pseudo-code of the mechanism TruPreTar is presented in Algorithm \ref{alg1}. In the mechanism, we maintain a variable of the remaining budget $B'$ which is initially set as $B$. Once a user $i$ is matched with payment $p_i$, we update $B'$ as $B'-p_i$. In addition, we update a decreasing global price $P$ during the algorithm process.
Intuitively, for each element in the iteration of the sorted list, if it is a task $j$, the task and its affiliated edges (the connected user should be unmatched) will be added to $G'$ if two conditions are satisfied after adding them: 1) $G'$ still has a right-perfect matching and 2) the upper bound of the payment for all tasks in $G'$ (i.e., $(|T'|+1)\cdot r_j$) is below the remaining budget. If task $j$ is added to $G'$, we update the global price $P$ as $r_j$.
If the element is a user $i$, then if $G'$ can maintain a right-perfect matching after discarding $i$ (i.e., $i$ is not \emph{critical} for the right-perfect matching of $G'$), we then remove her from $G'$, in this case, we also update the global price $P$ as $b_i$. Once the subgraph $G'$ is changed (either a task is added or a user is removed), we check if there are critical users for the right-perfect matching of $G'$, if so, for a critical user $i$, we allocate task $j$ in $G'$ to $i$ if $G'\backslash i\cup (i,j)$ has a right-perfect matching. The payment is set as the global price $P$ at this step, and then we update the remaining budget $B'$.

A walk-through example of TruPreTar is given in Figure \ref{fig:1}.

\begin{figure}[!tb]
\centerline{\includegraphics[width=0.65\columnwidth]{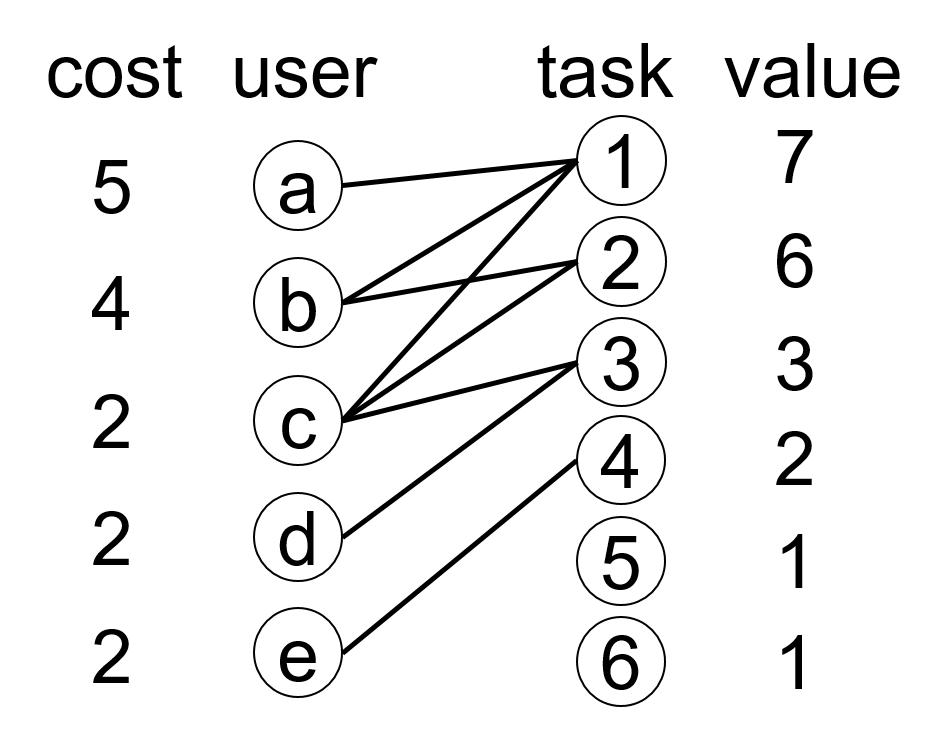}}
\caption{A walk-through example for our mechanism TruPreTar. Let the budget be $B=14$. Tasks 1 and 2 are first considered, and they are added to $G'$ with their affiliated edges since the budget is enough for them. So $G'$ has three users $a,b$ and $c$ and none of them is critical for the right-perfect matching of $G'$. In addition, the global price $P\leftarrow6$, and the remaining budget $B'$ is still $14$. Then, user $a$ is removed from $G'$ while updating the price $P\leftarrow 5$. Next, users $b$ and $c$ become critical for the right-perfect matching of $G'$ now, and assume they are matched with task $1$ and $2$, respectively. The payment to each of $b$ and $c$ is $P=5$, and hence we update $B'\leftarrow 4$. After that, task $3$ will be added to $G'$ with edge $(d,3)$, we have $P\leftarrow 3$, and obviously $d$ is critical and she will be allocated task $3$ with payment $3$. The remaining budget is also updated as $B'\leftarrow 1$. Finally, task 4 is considered but not added to $G'$ because of the budget constraint, and tasks 5 and 6 are skipped because they have no edges. Therefore, the output matching is $\{(b,1),(c,2),(d,3)\}$, and the revenue is 16 while the profit is 3.}
\label{fig:1}
\end{figure}

\section{Analysis of TruPreTar}

\begin{lem}
The mechanism TruPreTar satisfies incentive compatibility.
\end{lem}

\begin{proof}
Since each user has only one private value (i.e., cost), this is a single-parameter problem and hence we can use the Myerson's lemma:
\begin{lem}[\cite{myerson1981optimal}]
\label{lemma1}
In single parameter auctions, for a normalized mechanism $\mathcal{M}=(f, p)$, where $f$ is the allocation rule and $p$ is the payment rule, $\mathcal{M}$ is incentive compatible iff it satisfies:
\begin{enumerate}
  \item \textbf{Monotone allocation rule:} $\forall i \in N$, if $b'_i\le b_i$, then $i\in f(b_i, b_{-i})$ implies  $i\in f(b'_i,b_{-i})$ for every $c_{-i}$;
  \item \textbf{Threshold payment rule:} payment to each winning bidder is $\inf {\{b_i: i\notin f(b_i,b_{-i}) \}}$.
\end{enumerate}
\end{lem}
First, we prove the monotone allocation rule, i.e., once user $i$ is matched by bidding $b_i$, she must be matched by bidding $b'_i<b_i$. Let $e$ be the element that updates the global price $P$ as $p_i$.
When bidding $b_i$, we use $G^1_i$ to denote the graph $G'$ after $e$ is processed (i.e., the step that $P$ is updated as $p_i$), and $G^2_i$ has same definition while in the case of bidding $b'_i$. It's obvious that the algorithm process before $P$ is updated as $p_i$ is not affected by the bid of $i$. Therefore, we have $G^1_i=G^2_i$, and user $i$ is also critical for the right-perfect matching of $G^2_i$, hence the monotone allocation rule is satisfied.

Next, we prove the threshold payment rule, that is, if user $i$ bids any cost larger than the payment $p_i$, she will not be matched, otherwise, she will be matched with a task. If $b'_i < p_i$, it can be observed that the value of all tasks in $G^2_i$ is higher than $b'_i$, thus the edges of $i$ will not be deleted because $b'_i>r_j$. Therefore, similar to the proof of the monotone allocation rule, we have that $G^1_i=G^2_i$, and user $i$ is still critical, so she will be matched.
If $b'_i > p_i$, user $i$ will be considered before the element $e$. However, we can observe that, in the steps before element $e$ is processed, user $i$ is either not added to $G'$ or is not critical for the right-perfect matching of $G'$, otherwise she will be matched before element $e$. As a result, we have that user $i$ will be discarded if $b'_i > p_i$, and the threshold payment rule is satisfied.
\end{proof}
\begin{lem}
The mechanism TruPreTar is budget feasible.
\end{lem}
\begin{proof}
We can observe that once a task $j$ with value $r_j$ is added to $G'$, the payment to any user in $G'$ is no more than $r_j$, since the global price $P$ is non-increasing. Let $j'$ be the last task added to $G'$, $B'_{j'}$ the remaining budget before $j'$ is added, and $|T_{j'}|$ the number of tasks in $G'$ before adding $j'$. We have that $$\sum_{(i,j)\in M}p_i \le B- B'_{j'}+r_{j'}\cdot (|T_{j'}|+1) \le B$$ which concludes the proof.
\end{proof}

We omit the proofs of the following two lemmas due to space constraint, please refer to our full paper\footnote{https://arxiv.org/abs/1911.07706}.
\begin{lem}
The mechanism TruPreTar is individually rational for both users and platform.
\end{lem}
\begin{proof}
We need to prove that when matching a user $i$ with a task $j$, the payment to user $i$ should be larger than $b_i$ and smaller than $r_j$. It can be observed that, at any step during the algorithm process, the global price $P$ is always smaller than the values of all tasks in $G'$, and any tasks with lower value should not be considered until this step. Similarly, $P$ is also always larger than the costs of all users in $G'$, otherwise the user with higher cost should be either discarded or matched before. Therefore, we have that for any matched pair $(i,j)$, $b_i\le p_i \le r_j$, which concludes our proof.
\end{proof}

\begin{lem}
The mechanism TruPreTar satisfies computational efficiency.
\end{lem}
\begin{proof}
The subproblem of judging whether there is a right-perfect matching can be solved by Hungarian algorithm with complexity of $O((n+mn)^3)$. The loop complexity in the checking part (line \ref{checkstart}-\ref{checkend}) is at most $O(n^2m)$ The outermost complexity of our mechanism is $O(n+mn)$. So the overall complexity of our mechanism is polynomial.
\end{proof}

\begin{thm}
Our proposed mechanism TruPreTar is an incentive compatible, budget feasible, individually rational, and computational efficient mechanism.
\end{thm}

\section{Theoretical Guarantee on Revenue}\label{sec:lowerbd}

To show the theoretical guarantee of TruPreTar on revenue, we first introduce the large market assumption.
\begin{assumption}[Large Market Assumption]
\label{assm1}
We assume $c_i\ll B$ and $r_j\ll B$ for each user $i$ and each task $j$.
\end{assumption}
Intuitively, it's assumed that each individual user or task is negligible compared with the budget. This assumption is widely adopted in previous work \cite{vaze2017online,anari2014mechanism} and it is practical in real world as the revenue of a single ride is indeed very small.

Next, we prove the theoretical guarantee under tight budget.
\begin{thm}\label{tight}
Under the large market assumption, we have $\sum_{(i,j)\in M}r_j \ge B$ if the budget is tight.
\end{thm}
\begin{proof}
Let $j_1$ be the first task that is discarded because of budget constraint, $j0$ the last task added to $G'$ before $j1$, and $|T_{j0}|$ ($|T_{j1}|$) the number of tasks in $G'$ before considering $j0$ ($j1$). Assume that the set of tasks allocated between considering $j0$ and $j1$ is $A_T$, and the total payment for them is $P_A$. We have $r_{j0}\cdot(|T_{j0}|+1)\le B'_{j0}$ and $r_{j1}\cdot(|T_{j1}|+1) > B'_{j1}.$ Since the payments to users in $A_T$ are all between $r_{j0}$ and $r_{j1}$, we can get that
\begin{align*}
|T_{j1}|\cdot r_{j1} &\le |T_{j1}|\cdot r_{j0}\\
&= (|T_{j0}|+1)\cdot r_{j0}-|A_T|\cdot r_{j0} \\
&\le (|T_{j0}|+1)\cdot r_{j0}-P_A\\
&\le B'_{j0}-P_A\\
&=B'_{j1}
\end{align*}
Combining the above inequations, we have $$B'_{j1}-r_{j1}\cdot |T_{j1}|<r_{j1},$$ and further we obtain
\begin{align*}
\sum_{(i,j)\in M}r_j &\ge B-B'_{j1}+r_{j1}\cdot |T_{j1}|\\
&\ge B-r_{j1}\\
&\simeq B
\end{align*}
where the first inequation is because of the individual rationality of both platform and users and the last approximate equation is due to the large market assumption.
\end{proof}

Before we prove the theoretical guarantee under sufficient budget, we first demonstrate that a greedy algorithm as shown in Algorithm~\ref{alg2} has an approximation ratio of 2, i.e., it can achieve at least half of the optimal revenue under sufficient budget.

\begin{algorithm}[h]
\caption{A Greedy Mechanism}
\label{alg2}
\For {each task $j$ in decreasing order of $r_j$}{
\For {each edge $(i,j)$ of task $j$}{
\If {user $i$ is not matched}{
Match $i$ with $j$.

Skip to next task.
}}}
\end{algorithm}

\begin{lem}
\label{2app}
Algorithm~\ref{alg2} is a 2-approximation algorithm if the budget is sufficient.
\end{lem}

The proof is provided in our full paper due to space limitation.

Next, we prove the following lemma by showing that the allocation of our mechanism coincides with a particular run of the greedy algorithm.
\begin{thm}\label{sufficient}
The mechanism TruPreTar is a 2-approximation mechanism if the budget is sufficient.
\end{thm}
\begin{proof}
Since the budget is sufficient, we know that no task is discarded due to the budget limitation. We denote the matching in our mechanism as $M$. It's assumed that this matching is produced as following: For each task $j$ in decreasing order of $r_j$, if $j\in M$, we allocate task $j$ to its matched user in $M$, otherwise the task is skipped. Now we prove that, in the greedy algorithm, this process can also happen.

It's obvious that for task $j \in M$, we can assign task $j$ to its matched user in $M$ in the greedy algorithm, thus we only need to prove that for each task $j \not\in M$, when we consider it in the greedy algorithm, there is no unassigned user that has an edge to $j$.

Next, for contradiction, assume that we can find such an unmatched user $i$ that has an edge to $j \not\in M$ when processing $j$ in the greedy algorithm. Then in our mechanism, when considering $j$, there will be a right-perfect matching for graph $G'\cup \{j\}$, i.e., $M\cup (i,j)$. Thus, the task will be added to $G'$. Note that in our mechanism, any task added to $G'$ will end up being matched. This contradicts with our assumption and hence the lemma is proved.
\end{proof}

To understand how ``good'' the approximation ratio of 2 is, next we prove that the lower bound is at least $\sqrt{2}$.
\begin{thm}\label{thm:lowerbd}
There is no mechanism that satisfies incentive compatibility and individual rationality can achieve better than $\sqrt{2}$-approximation when the budget is sufficient.
\end{thm}
\begin{proof}
We prove the lemma with a concrete counter example. Assume there exists a mechanism $\mathscr{F}$ that can achieve an approximation ratio better than $\sqrt{2}$. Consider case 1 where there are two users $\{a, b\}$, three tasks $\{1, 2, 3 \}$ and 4 edges $\{(a,1),(a,2),(b,2),(b,3)\}$. In addition, we have that $c_a=\epsilon, c_b=\sqrt{2}+1, r_1=1+\epsilon, r_2=\sqrt{2}+1+\epsilon, r_3=\sqrt{2}+1$, where $\epsilon$ is a small positive number. We can observe that the optimal matching should be $\{(a,2),(b,3)\}$ which achieves revenue of $2\sqrt{2}+2+\epsilon$, we now prove that any mechanism that satisfies the above properties can achieve total revenue of at most $2+\sqrt{2}+2\epsilon$.

First we consider case 2 where the only difference with case 1 is that $c_b=\sqrt{2}+1+\frac{\epsilon}{2}$ and hence the edge $(b,3)$ has to be deleted. In case 2, mechanism $\mathscr{F}$ can only output the matching $\{(a,1),(b,2)\}$, otherwise, the approximation can be at least$\frac{r_1+r_2}{r_2}>\sqrt{2}$. Moreover, due to individual rationality, $p_b\ge \sqrt{2}+1+\frac{\epsilon}{2}$.

Then we consider case 1, in the output matching of $\mathscr{F}$, if user $b$ is matched with task $3$, the payment is at most $1+\sqrt{2}$ and the utility of user $b$ is at most 0. If user $b$ misreports cost of $\sqrt{2}+1+\frac{\epsilon}{2}$, it becomes case 2, and as stated above, the utility of user $b$ can be at least $\frac{\epsilon}{2}$. Thus, for incentive compatibility, mechanism $\mathscr{F}$ has to allocate task $2$ to user $b$ and pay at least $\sqrt{2}+1+\frac{\epsilon}{2}$. Hence the output matching of $\mathscr{F}$ can achieve revenue of at most $2+\sqrt{2}+2\epsilon$, and the approximation ratio limit is $\sqrt{2}$ when $\epsilon\rightarrow 0$.
\end{proof}

{\bf {Practical Guideline:}} We provide a guideline as to how the above theoretical results can be applied to practice. The platform can initially set a sufficient budget and gain a revenue $R_{suf}$. Based on Theorem \ref{sufficient}, we know that the optimal revenue is at most $2\cdot R_{suf}$. Hence after that, the platform can set a tighter budget $B=\beta\cdot R_{suf}$ where $0<\beta\le 2$, and Theorem \ref{tight} guarantees that TruPreTar will achieve at least $\beta/2$ of the optimal solution's revenue. This way, the platform can control the budget while ensuring a minimum revenue.

\section{Evaluation}

We conduct simulation using a real-world dataset obtained from a large bike-sharing company in China called Mobike. The bike riding data cover $8\times 8$ regions of Beijing with each region being 0.6km $\times$ 0.6km, and are dated from May 10th to 14th, 2017. With the same distribution of destinations in this dataset, we build a simulator which can randomly generate users' destinations. In the experiments, we set the number of users $n=200$, and test different location numbers $m$. The cost of each user $c_i$ is drawn from uniform distribution over $[0,\overline{c}]$ where $\overline{c}=5$. The value of a task is calculated as the difference between the Kullback-Leibler (KL) divergences \cite{kullback1951information} before and after fulfilling the task, similarly to previous work \cite{pan2019deep,lv2019hardness}. Because of the space limitation, we refer the authors to \cite{lv2019hardness} for concrete calculation of the task value. The acceptable range $h$ is set as 300m and 600m, respectively. We also test the budget of 50 and 500 where 500 is sufficient while 50 is not. In addition, we conduct each experiment 10 times and take the average.

We compare TruPreTar with the following mechanisms:

\begin{itemize}
\item {\bf APP-OPT:} As the budgeted matching problem is an NP-hard problem (it can be reduced to the knapsack problem), we cannot give an optimal allocation as benchmark. However, if the budget is tight and the large market assumption holds simultaneously, the following strategy is approximately an optimal mechanism: consider edges in decreasing order of $r_j/b_i$, match user $i$ and task $j$ if they are not ever matched before, and pay user $i$ exactly her bid. The process stops until the budget is exhausted or there are no more edges.
    This mechanism achieves the maximum revenue under the above two assumptions but it's not truthful.
\item {\bf Greedy:} The greedy mechanism considers users in increasing order of their bids and allocates the task with the highest value as long as it is higher than the bid of the user. Once a user is matched, the price for all winning users is updated as the bid of her next user (whose bid is higher). The mechanism stops once a user cannot find a feasible task or the overall payment is above the budget, and pays all the winning users the bid of this unmatched user. If all users can be matched, we don't match the last one, and set her bid as the price for all other users. This greedy mechanism is incentive compatible but cannot provide any theoretical guarantee.
\item {\bf Surge:} Surge pricing is widely used in practice and it's an effective way to promote the revenue of platform \cite{guda2019your}. We adopt a simple version of surge pricing here. Consider users in increasing order of their bids and allocate the task with the highest value to them once $\alpha r_j$ is higher than the bid of the user, where $\alpha\in(0,1)$ is a surge factor and $r_j$ is the value of the task. The user is paid $\alpha r_j$. The algorithm also terminates until the budget is exhausted or there are no feasible users. In the experiments, we set $\alpha=0.8$. It can be easily shown that the Surge mechanism is not truthful.
\end{itemize}
\begin{figure}[!tb]
\subcaptionbox{Revenue for $B=50$}{\includegraphics[width=0.49\columnwidth]{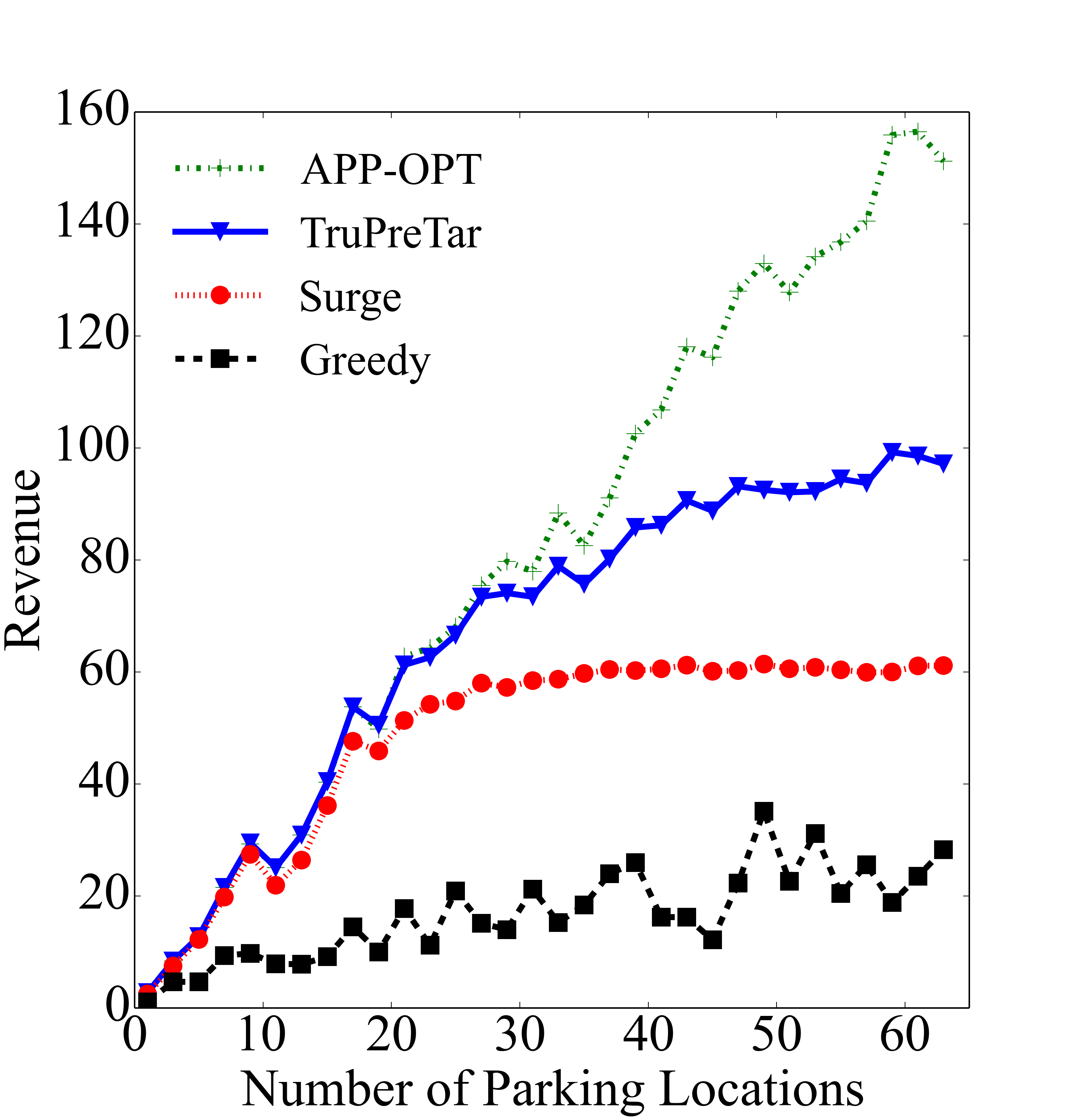}}
\hfill
\subcaptionbox{Profit for $B=50$}{\includegraphics[width=0.49\columnwidth]{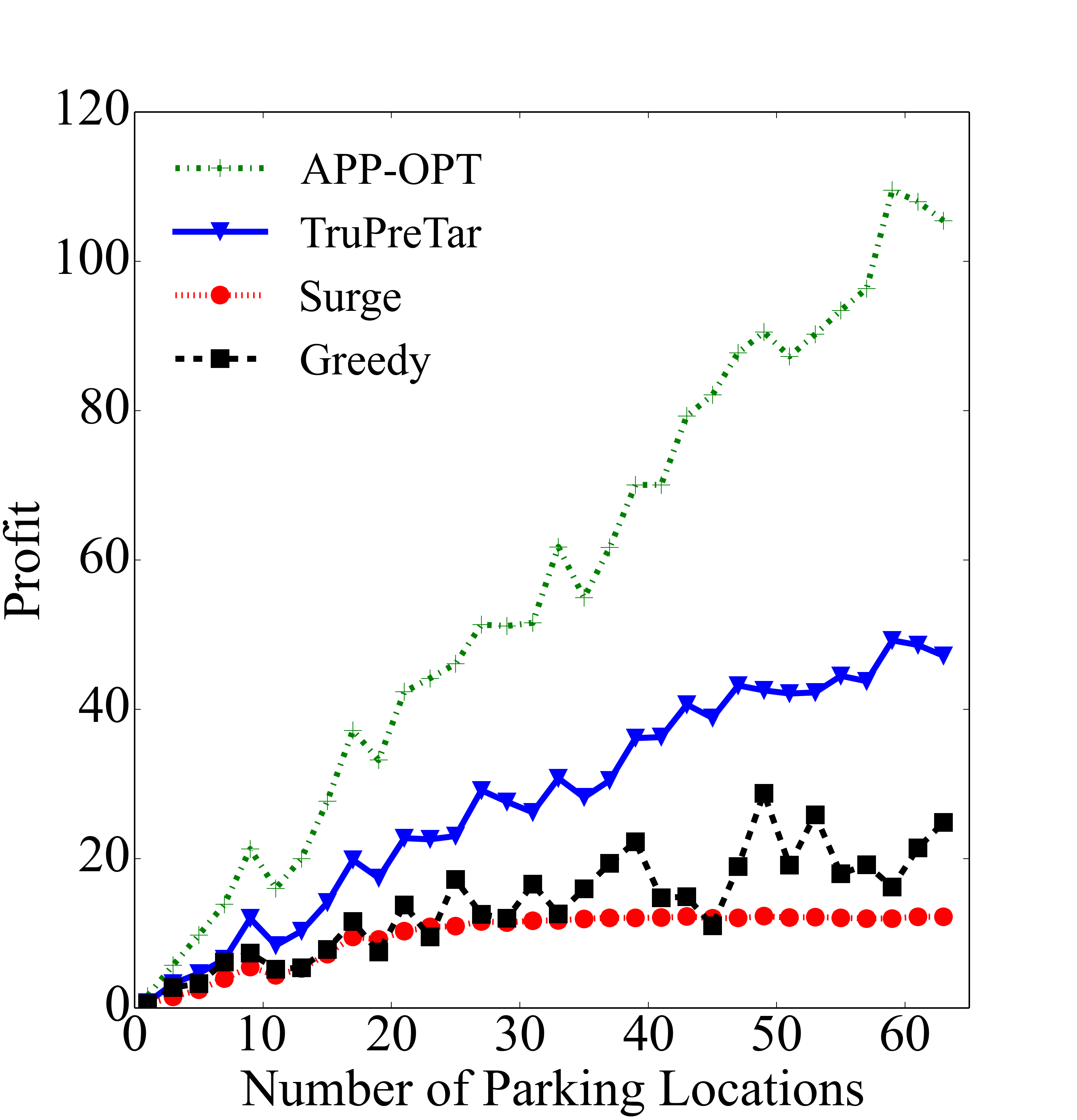}}
\vfill
\subcaptionbox{Revenue for $B=500$}{\includegraphics[width=0.49\columnwidth]{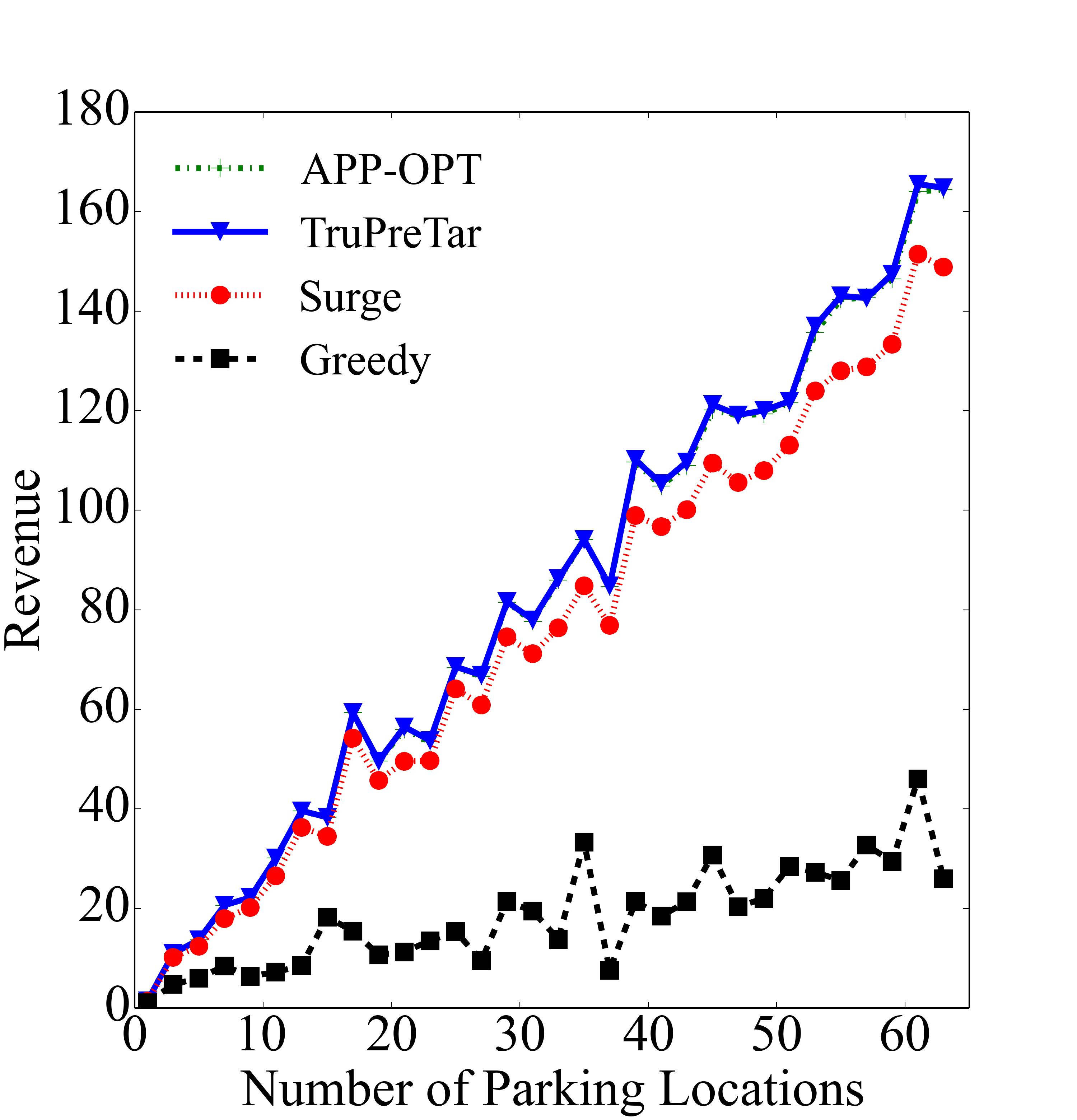}}
\hfill
\subcaptionbox{Profit for $B=500$}{\includegraphics[width=0.49\columnwidth]{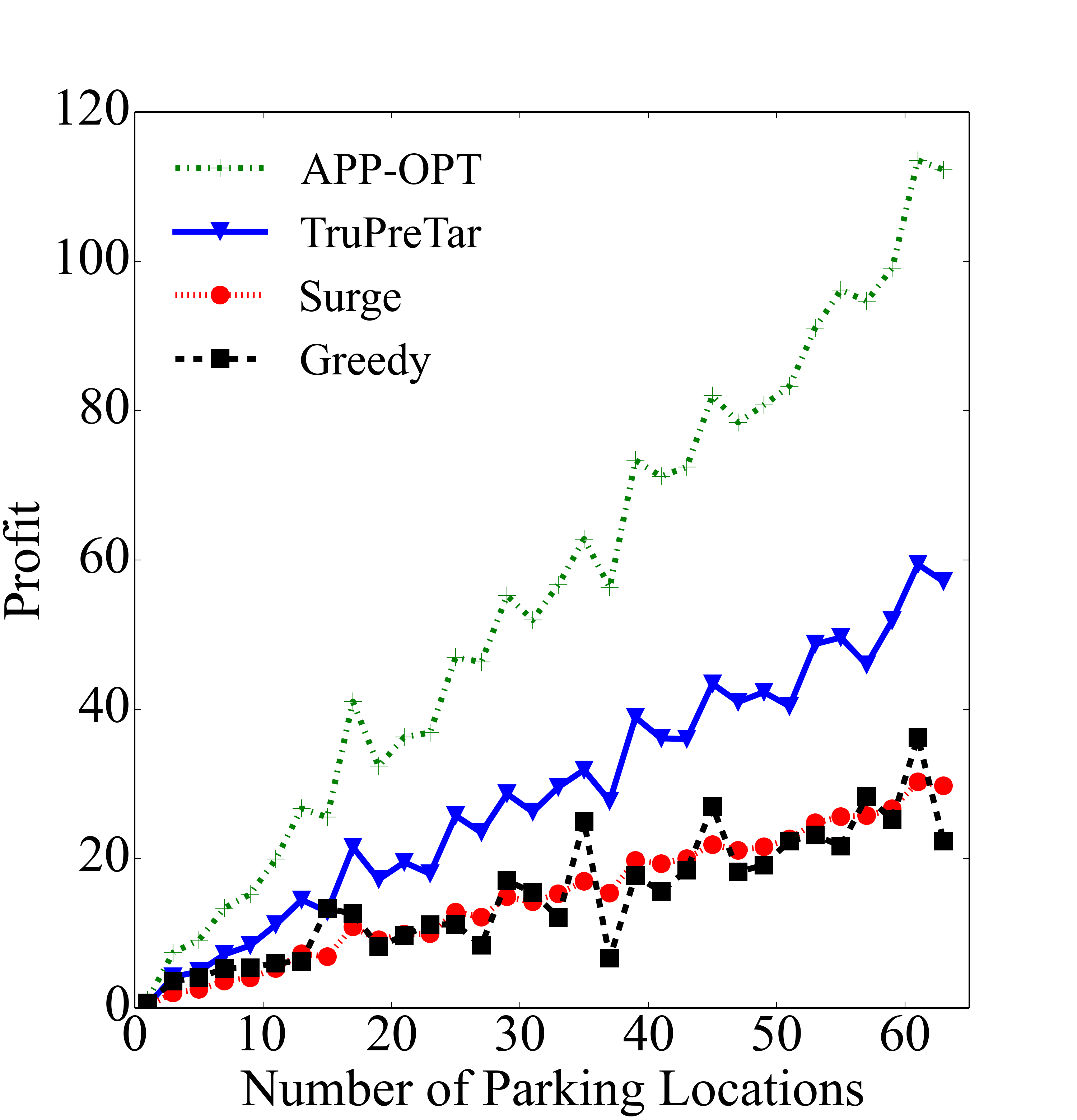}}
\caption{Revenue and profit comparison when $h=300m$.}
\label{fig:300}
\end{figure}

\begin{figure}[!tb]
\subcaptionbox{Revenue for $B=50$}{\includegraphics[width=0.49\columnwidth]{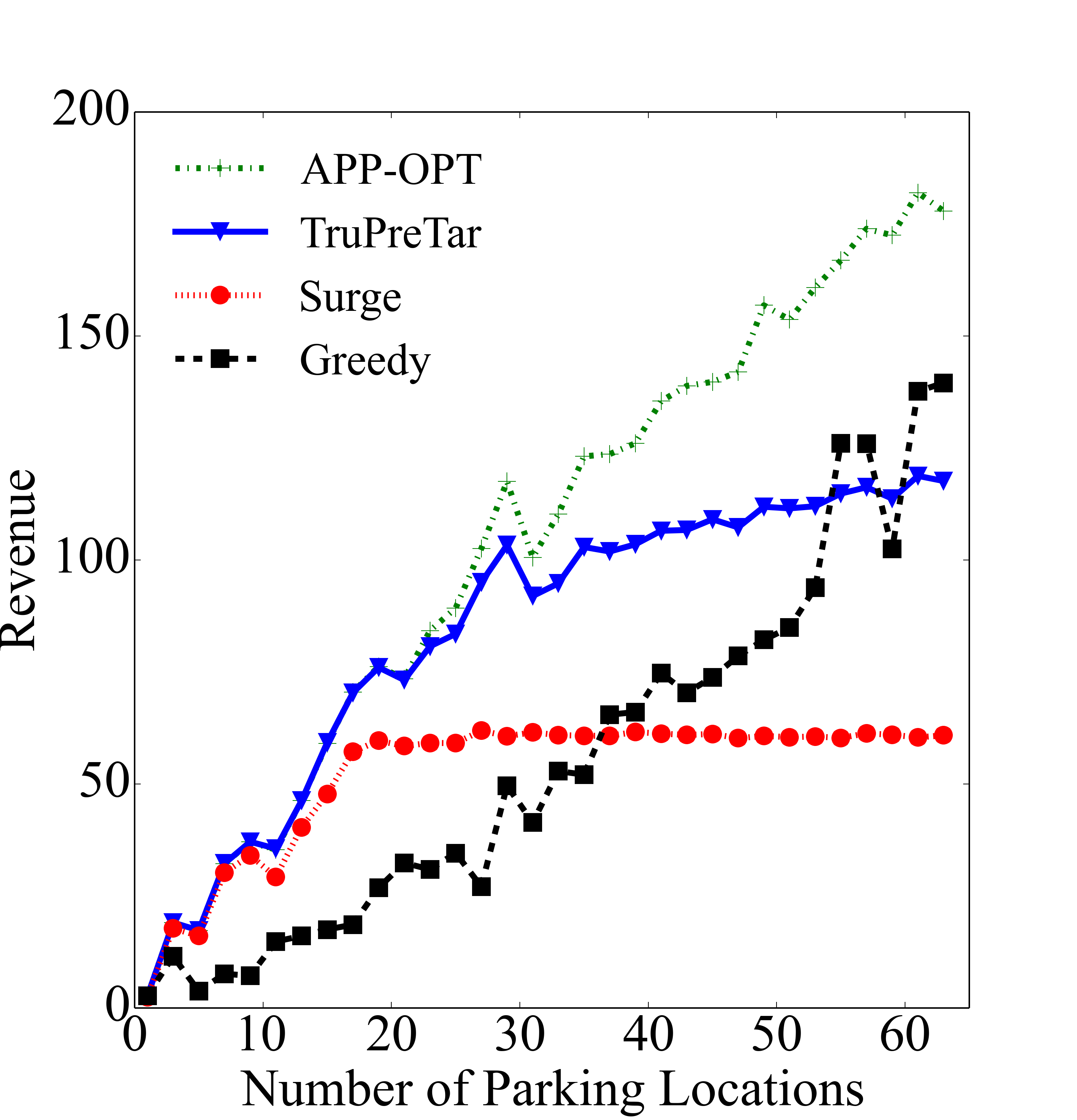}}
\hfill
\subcaptionbox{Profit for $B=50$}{\includegraphics[width=0.49\columnwidth]{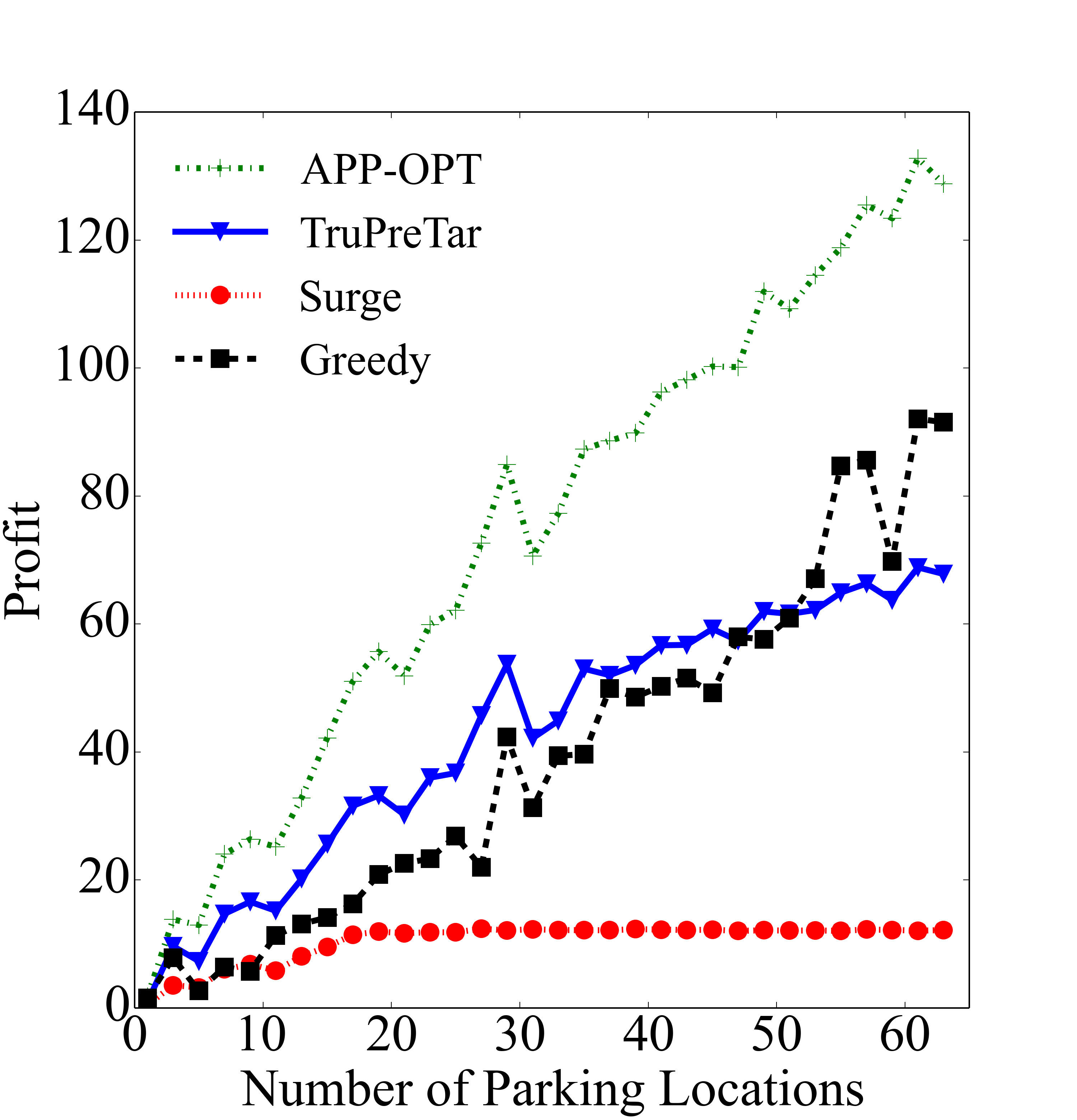}}
\vfill
\subcaptionbox{Revenue for $B=500$}{\includegraphics[width=0.49\columnwidth]{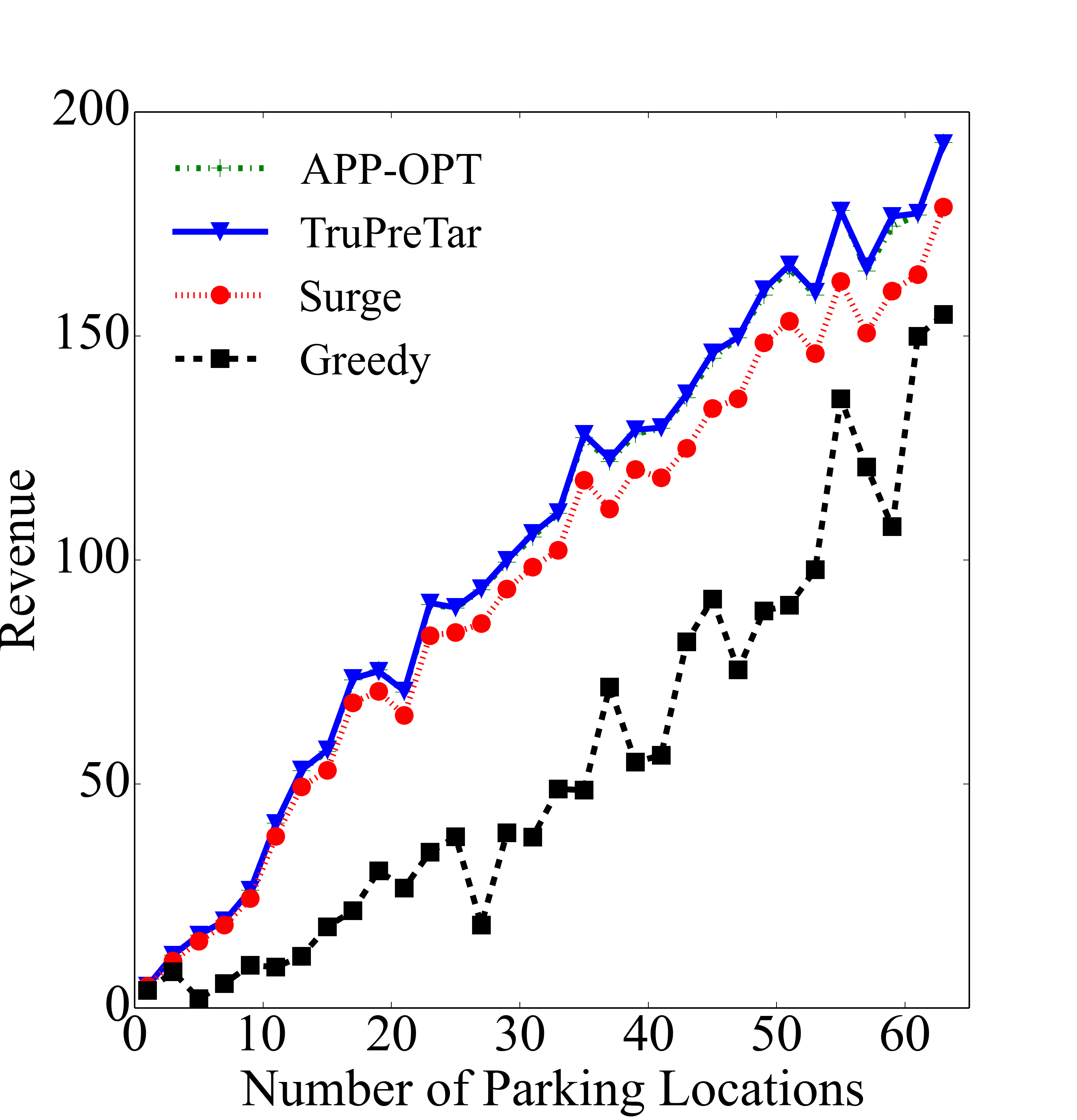}}
\hfill
\subcaptionbox{Profit for $B=500$}{\includegraphics[width=0.49\columnwidth]{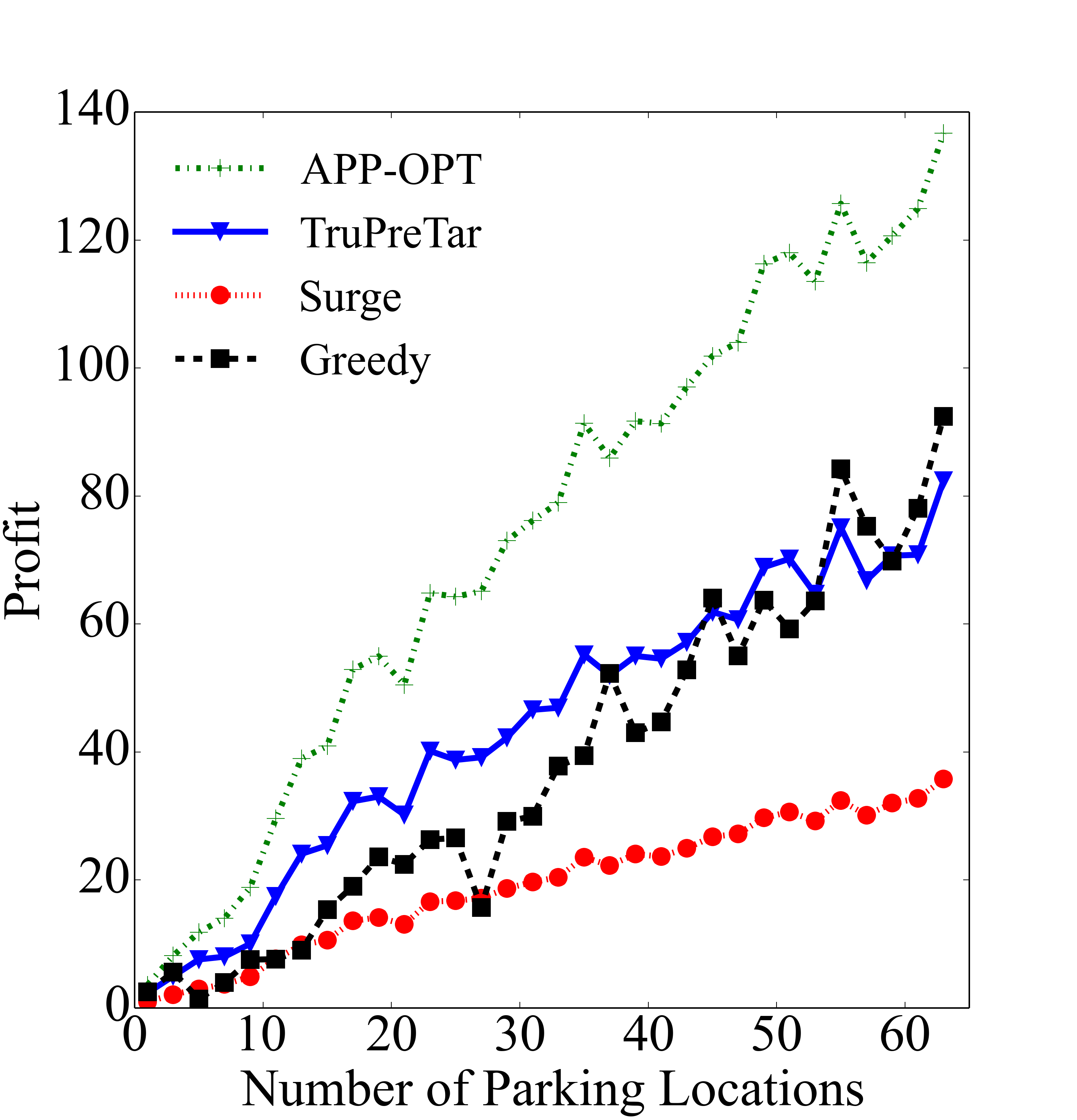}}
\caption{Revenue and profit comparison when $h=600m$.}
\label{fig:600}
\end{figure}

\subsection{Results}

The experiment results for different acceptable extra distance $h=300m$ and $h=600m$ are given in Figure \ref{fig:300} and Figure \ref{fig:600}, respectively.
In Figure \ref{fig:300}, we show the revenue and profit of platform under different parking location number $m$ and different budget $B$.
In general, as $m$ increases, revenue and profit both increase, because of more right nodes in the bipartite graph. When the budget is tight ($B=50$), we make the following specific observations. The revenue and profit of our mechanism are constantly higher than Surge and Greedy. When $m$ is less than 30, as the bipartite graph is small, the budget is enough and our mechanism achieves nearly the same revenue as APP-OPT. When $m$ is larger than 30, the budget will be exhausted and the revenue (and hence profit) of Surge mechanism stops increasing because it always has a constant ratio between payment and revenue. However, TruPreTar can output a better matching when the bipartite graph is larger and it shows much better performance over Surge and Greedy on both revenue and profit.
When the budget is sufficient, our mechanism is not as good as APP-OPT but outperforms the others. As APP-OPT pays users exactly their costs, the mechanism is budget-saving but not practical in application because users will bid higher costs.

In Figure \ref{fig:600}, $h$ is larger so more edges appear in the bipartite graph. In this case, the performance of APP-OPT, TruPreTar and Surge are similar to Figure \ref{fig:300}, but Greedy shows both much better revenue and profit. This is due to the fragility of Greedy mechanism as a single user can determine the termination of the mechanism. Thus, when $B=300m$, the mechanism stops quickly because of the lack of edges, and when $B=600m$, the termination is delayed. To show this, we test another 100 rounds of $m=60$ for both Greedy and TruPreTar, and the maximum revenue of Greedy is 161.3, the minimum revenue is 19.7, and the variance is 648.1, while the data for TruPreTar is 127.2, 98.8 and 32.2 respectively. The profit shows similar results. Therefore, although Greedy can sometimes be slightly better than our mechanism, it is fragile and has much larger fluctuation.

In summary, our mechanism TruPreTar performs well on both revenue and profit, on top of its several desirable theoretical properties.

\section{Conclusion and Future Work}

In this paper, we have proposed an incentive mechanism to solve the bike rebalancing problem with predicted task value for bike sharing systems. This mechanism, called TruPreTar, satisfies incentive compatibility, budget feasibility, individual rationality, and computational efficiency. It also provides theoretical guarantees on company revenue under different budget constraints.
Its performance was evaluated using simulations based on real-world data, and the results demonstrate its superiority in terms of both revenue and profit. In future work, we will extend our algorithm into a real-time decision-making mechanism, and conduct pilot studies in real cities.

\section{Acknowledgments}

This work was supported in part by Science and Technology Innovation 2030 ``New Generation Artificial Intelligence'' Major Project No. 2018AAA0100905, in part by China NSF grant No. 61972252, 61972254, 61672348, and 61672353, in part by Joint Scientific Research Foundation of the State Education Ministry No. 6141A02033702, and in part by Alibaba Group through Alibaba Innovation Research Program. The opinions, findings, conclusions, and recommendations expressed in this paper are those of the authors and do not necessarily reflect the views of the funding agencies or the government.
\bibliographystyle{aaai}
\bibliography{AAAI-LvH.3279}
\end{document}